\documentclass[11pt]{article}

\usepackage{amsmath,amssymb,amsbsy,amsfonts,amsthm,latexsym,
            amsopn,amstext,amsxtra,euscript,amscd,color}
            
\PassOptionsToPackage{hyphens}{url}\usepackage{hyperref}
    
\usepackage[capbesideposition=outside,capbesidesep=quad]{floatrow}

\newtheorem{theorem}{Theorem}[section]
\newtheorem{corollary}{Corollary}[section]

\newtheorem{thm}{Theorem}[section]

\newcommand{\Tr}{{\rm Tr}}
\newcommand{\Trn}{{\rm Tr}_n}

\def\+{\oplus}

\def\F{{\mathbb F}}

\def\F{{\mathbb F}}

\def\00{{\bf 0}}
\def\11{{\bf 1}}
\def\+{\oplus}

\def\\{\cr}
\def\({\left(}
\def\){\right)}

\newcommand{\cardinality}[1]{\# #1}

\usepackage[colorinlistoftodos,prependcaption,textsize=tiny]{todonotes}

\providecommand{\newoperator}[3]{%
  \newcommand*{#1}{\mathop{#2}#3}}

\newoperator{\FD}{\mathrm{FD}}{\nolimits}

\begin{document}
\title{\bf The $c$-differential spectrum of $x\mapsto x^{\frac{p^n+1}{2}}$ in finite fields of odd characteristics}
%\acknowledgement{bb}
% \author{Pantelimon~St\u anic\u a et al.%\thanks{$^*$Corresponding author}
%\thanks{P. St\u anic\u a is with Applied Mathematics Department, Naval Postgraduate School, Monterey, CA 93943, USA. %\protect\\ 
%E-mail: pstanica@nps.edu}
\author{\Large Constanza Riera$^1$,   Pantelimon St\u anic\u a$^2$\footnote{Corresponding author}, \, Haode Yan$^3$\\
\vspace{.01cm}\\
\small $^1$Department of Computer Science, Electrical Engineering, \\
\small and Mathematical Sciences,\\
\small  Western Norway University of Applied Sciences, \\
\small 5020 Bergen, Norway; {\tt csr@hvl.no}\\ 
\small $^2$ Department of Applied Mathematics, Naval Postgraduate School,\\
\small Monterey, CA 93943-5212, U.S.A.; {\tt pstanica@nps.edu}\\
\small $^3$ School of Mathematics, Southwest Jiaotong University,\\
\small Chengdu 610031, China; {\tt hdyan@swjtu.edu.cn}
}
\date{}

\maketitle

\begin{abstract}
In the paper, we concentrate on  the map $x\mapsto x^{\frac{p^n+1}{2}}$ on $\F_{p^n}$ and using combinatorial and number theory techniques, we compute its detailed $c$-differential spectrum for all values of $c\neq 1$ (the spectrum for $c=1$ is known).  
\end{abstract}


\section{Introduction and basic definitions}
 
 We will introduce here only some basic notations and definitions on Boolean and $p$-ary functions (where $p$ is an odd prime); the reader can consult~\cite{Bud14,CarletBook21,CS17,MesnagerBook,Tok15}, for more on these objects and their cryptographic properties.

For a positive integer $n$ and $p$ a prime number, we denote by $\F_p^n$ the $n$-dimensional vector space over $\F_p$, and by $\F_{p^n}$ the  finite field with $p^n$ elements, while $\F_{p^n}^*=\F_{p^n}\setminus\{0\}$ will denote the multiplicative group. For $a\neq 0$, we often write $\frac{1}{a}$ to mean the inverse of $a$ in the multiplicative group of the  finite field under discussion.  
We use $\cardinality{S}$ to denote the cardinality of a set $S$.
% and $\bar z$, for the complex conjugate.
We call a function from $\F_{p^n}$ (or $\F_p^n$) to $\F_p$  a {\em $p$-ary  function} on $n$ variables. For positive integers $n$ and $m$, any map $F:\F_{p^n}\to\F_{p^m}$ (or, $\F_p^n\to\F_p^m$)  is called a {\em vectorial $p$-ary  function}, or {\em $(n,m)$-function}. 
When $m=n$, $F$ can be uniquely represented as a univariate polynomial over $\F_{p^n}$ (using some identification, via a basis, of the finite field with the vector space) of the form
$
F(x)=\sum_{i=0}^{p^n-1} a_i x^i,\ a_i\in\F_{p^n},
$
whose {\em algebraic degree}   is then the largest Hamming weight of the exponents $i$ with $a_i\neq 0$. To (somewhat) distinguish between the vectorial and single-component output, we shall use upper/lower case to denote the functions.
 
 As customary, we will use the following concepts:
 \begin{enumerate}
 \item 
 The absolute trace $\Trn$ and the relative trace $\Tr_{\F_{p^n}/\F_{p^m}}$, defined as $\Tr_{\F_{p^n}/\F_{p^m}} (x)=\sum_{i=0}^{\frac{n}{m}-1} x^{p^{mi}}$.
 \item 
Given a $p$-ary  function $f:\F_{p^n}\to\F_{p^n}$, the derivative of $f$ with respect to~$a \in \F_{p^n}$ is the $p$-ary  function
$
 D_{a}f(x) =  f(x + a)- f(x), \mbox{ for  all }  x \in \F_{p^n}.
$ 
\item
For an $(n,m)$-function  
$F$, and $a\in\F_{p^n},b\in\F_{p^m}$, we let $\Delta_F(a,b)=\cardinality{\{x\in\F_{p^n} : F(x+a)-F(x)=b\}}$. We call the quantity
$\delta_F=\max\{\Delta_F(a,b)\,:\, a,b\in \F_{p^n}, a\neq 0 \}$ the {\em differential uniformity} of $F$. If $\delta_F= \delta$, then we say that $F$ is differentially $\delta$-uniform. If $m=n$ and $\delta=1$, then $F$ is called a {\em perfect nonlinear} ({\em PN}) function, or {\em planar} function. If  $m=n$ and $\delta=2$, then $F$ is called an {\em almost perfect nonlinear} ({\em APN}) function. It is well known that PN functions do not exist if $p=2$. 

\item
 As defined in~\cite{EFRST20}, for an $(n,m)$-function $F$, $a\in\F_{p^n}$ 
 and $c\in\F_{p^m}$, the ({\em multiplicative}) {\em $c$-derivative} of $F$ with respect to~$a \in \F_{p^n}$ is the  function
\[
 _cD_{a}F(x) =  F(x + a)- cF(x), \mbox{ for  all }  x \in \F_{p^n}.
\]
We let the entries of the $c$-Difference Distribution Table ($c$-DDT) be defined by ${{_c}\Delta}_F(a,b)=\cardinality{\{x\in\F_{p^n} : F(x+a)-cF(x)=b\}}$. We call the quantity
\[
{_c}\Delta_F=\max\left\{{{_c}\Delta}_F(a,b)\,|\, a\in \F_{p^n}, b\in\F_{p^m} \text{ and } a\neq 0 \text{ if $c=1$} \right\}\]
the {\em $c$-differential uniformity} of~$F$.  

\item 
 If $\delta_{F,c}=\delta$, then we say that $F$ is differentially $(c,\delta)$-uniform (or that $F$ has $c$-uniformity $\delta$, or that %for short, 
{\em $F$ has $\delta$-uniform $c$-DDT}). If $\delta=1$, then $F$ is called a {\em perfect $c$-nonlinear} ({\em PcN}) function (certainly, for $c=1$, they only exist for odd characteristic $p$; however, as shown in~\cite{EFRST20}, there exist PcN functions for $p=2$, for all  $c\neq1$). If $\delta=2$, then $F$ is called an {\em almost perfect $c$-nonlinear} ({\em APcN}) function. 
When we need to specify the constant $c$ for which the function is PcN or APcN, then we may use the notation  $c$-PN, or $c$-APN.
Note that if $F$ is an $(n,n)$-function,   then $F$ is $c$-PN if and only if $_cD_a F$ is a permutation polynomial.

\item
As in \cite{YZ22}, let $F:\F_{p^n}\to\F_{p^n}$ be a power function $F(x)=x^d$ with $c$-differential uniformity ${_c}\Delta_F$. Denote by ${_c}\omega_i=\cardinality\{b\in\F_{p^n}: {_c}\delta_F(1,b)=i\}$.  We define the {\em $c$-differential spectrum}  of $F$ as the multiset $\mathbb{S}=\{{_c}\omega_i: 0\leq i\leq {_c}\Delta_F \mbox{ and } {_c}\omega_i>0\}$.

NB: Note that, for $a=0$, for a power function $F(x)=x^d$, where $k=\gcd(d,p^n-1)$, the equation $F(x+a)-cF(x)=b$ becomes $(1-c)x^d=b$. If $k=1$, the equation has a unique solution. If $k>1$, and $b/(1-c)=g^s$, for some $s$, where $g$ is a primitive element of the field, the equation has solutions (namely $k$ of them) if and only if $k\,|\,s$. In the computation of the differential or $c$-differential uniformity for a power frunction,  for $a\neq0$,  it is customary to replace  $a$ by $1$, since an equation of the form $(x+a)^d-cx^d=b$ is equivalent to $(y+1)^d-c'y^d=b'$, where $y=x/a$, $b'=b/a^d, c'=c/a^d$.

	In \cite[Theorem 4]{YZ22}, it was proved  that for  any monomial function $x^d$, for $c\neq 1$, 
	\[	\sum_{i=0}^{{}_c{\Delta _F}}{}_c\omega_i=\sum_{i=0}^{{}_c{\Delta _F}}(i\cdot{}_c\omega_i)=p^n\]
	and
	 $$\sum_{i=0}^{_c\Delta_F}(i^2\cdot _c\omega_i)=\frac{_cN_4-1}{p^n-1}-\gcd(d,p^n-1),$$
where
 \begin{equation*}\label{equationsystem}
	{}_cN_{4}=\# \left\{ {\left( {{x_1},{x_2},{x_3},{x_4}} \right) \in (\F_{q})^4 : \begin{cases}
				{x_1} - {x_2} + {x_3} - {x_4} = 0\\
				x_1^d - cx_2^d + cx_3^d - x_4^d = 0
	\end{cases} }  \right\}. 
\end{equation*} 
This gives the relationship between the $c$-differential spectrum and the number of solutions of a certain system of equations.

\item
We will denote by $\eta(\alpha)$ the quadratic character of $\alpha$ (that is, $\eta(\alpha)=0$ if $\alpha=0$, $\eta(\alpha)=1$ if $0\neq \alpha$ is a square, $\eta(\alpha)=-1$ if $\alpha$ is not a square).
\end{enumerate}

In this paper we are computing precisely the $c$-differential spectrum of the function $F(x)=x^{\frac{p^n+1}{2}}$ over $\F_{p^n}$. Since some of the proofs are rather long, we will split the analysis in three results, treating the cases $c=0$, $c=-1$ and $c\not= 0,\pm 1$, separately. We also provide a result for the particular case $c^2=-1$. 
Recall that the spectrum  (and differential uniformity) for $c=1$ is known for this map~\cite{CHNC13}. As a byproduct, several character sums are explicitly found.

\section{The spectrum of $x^{\frac{p^n +1}{2}}$, $p$ odd prime}

Since the case of  $c=1$ was treated in~\cite{CHNC13}, for this function,  we assume below that $c\neq 1$.
We include here the $c$-differential uniformity for this function, and we give in the proof a complete characterization of the conditions on $b$ to have any number of solutions for the equation $(x+1)^{\frac{p^n +1}{2}}-cx^{\frac{p^n +1}{2}}=b$, where $c$ varies through the field. We note that this case is not included in~\cite{YZ22}, who treated the $(-1)$-spectrum for the function $x^{\frac{p^k +1}{2}}$, under the assumption $\gcd(n,k)=1$ and $\frac{2n}{\gcd(2n,k)}$ being even.

We will use the following result (see \cite{CHNC13}) about the cyclotomic numbers: Let $S_{i,j}=\{x\neq0,-1: \eta(x+1)=i,\eta(x)=j\}$, where $i,j\in\{\pm1\}$. Then, if $\eta(-1)=1$ (which happens if $p^n\equiv1 \mod 4$), then $|S_{1,1}|=\frac{p^n-5}{4}$ and $|S_{-1,1}|=|S_{1,-1}|=|S_{-1,-1}|=\frac{p^n-1}{4}$. If $\eta(-1)=-1$ (which happens if $p^n\equiv3\mod 4$), then $|S_{1,-1}|=\frac{p^n+1}{4}$ and $|S_{1,1}|=|S_{-1,1}|=|S_{-1,-1}|=\frac{p^n-3}{4}$. %These cardinalities are known as the {\em cyclotomic numbers}.

We start with the easy case of $c=0$.
\begin{theorem} 
Let $F(x)=x^d$, where $d=\frac{p^n +1}{2}$ (for $p$ odd).  If $c=0$, then the function is PcN if $\eta(-1)=1$, and APcN if $\eta(-1)=-1$, with spectrum $\left\{_c{\omega}_0=\frac{p^n-1}{2},_c{\omega}_1=1,_c{\omega}_2=\frac{p^n-1}{2}\right\}$. 
\end{theorem}

\begin{proof} Let $d=\frac{p^n +1}{2}=\frac{p^n-1}{2}+1$.  We look  at the equation
$$F(x+1)-cF(x)=b,$$
and so,
\begin{align}
(x+1)^{\frac{p^n +1}{2}}-cx^{\frac{p^n +1}{2}} %\nonumber \\
&=(x+1)^{\frac{p^n -1}{2}+1}-cx^{\frac{p^n -1}{2}+1}\nonumber \\
&=\eta(x+1)(x+1)-c\eta(x)x=b.\label{eq:main}
\end{align}

\noindent
When  $c=0$, Equation~\eqref{eq:main} becomes  $\eta(x+1)(x+1)=b$.

\begin{itemize}
\item If $x=-1$, then, $b=0$.
\item Let $\eta(x+1)=1$. Then, the equation becomes $x+1=b$, so $x=b-1$. This solution is valid if and only if $\eta(b)=1$.
\item Let $\eta(x+1)=-1$. Then, the equation becomes $-x-1=b$, so $x=-b-1$. This solution is valid if and only if $\eta(-b)=-1$.
\end{itemize}
Therefore, if $\eta(-1)=1$, then we have a unique solution for every $b$. The function is then 0-PN (that is, a permutation), so its spectrum is, naturally,
\begin{equation}
\label{eq:spectrum1}
\{{_0}\omega_1=p^n\}.
\end{equation}
If $\eta(-1)=-1$, then we have a single solution for $b=0$, two solutions for $b$ such that $\eta(b)=1$, and no solutions otherwise. The function is then 0-APN, with spectrum 
\begin{equation}
\label{eq:spectrum2}
\left\{_0{\omega}_0=\frac{p^n-1}{2},\,_0{\omega}_1=1,\,_0{\omega}_2=\frac{p^n-1}{2}\right\}.
\end{equation}
The claim is therefore shown.
\end{proof}

We continue with the more complicated case of $c=-1$.
\begin{theorem} 
Let $F(x)=x^d$, where $d=\frac{p^n +1}{2}$ (for $p$ odd), over $\F_{p^n}, n\geq 1$. If $c=-1$, then the spectrum is $\left\{_c{\omega_0}=\frac{p^n-1}{2},\,_c{\omega_1}=\frac{p^n-3}{2},\,_c{\omega_{\frac{p^n+3}{4}}}=2\right\}$ if $\eta(-1)=1$, and $\left\{_c{\omega_0}=\frac{3p^n-5}{4},\,_c{\omega_2}=\frac{p^n-3}{4},\,_c{\omega_{\frac{p^n+1}{4}}}=1,\,_c{\omega_{\frac{p^n+5}{4}}}=1\right\}$, if $\eta(-1)=-1$.  
\end{theorem}
\begin{proof}
When $c=-1$, Equation~\eqref{eq:main} becomes 
 \begin{equation}
 \label{eq:c-1}
 \eta(x+1)(x+1)+\eta(x)x=b.
 \end{equation}
Assuming $\eta(-1)=1$, we get the solutions:
\begin{itemize}
\item for $b=1$, $x=0$ and all elements $x\in\F_{p^n}$ such that $\eta(x+1)=1=-\eta(x)$ (that is, $|S_{1,-1}|=\frac{p^n-2+\eta(-1)}{4}=\frac{p^n-1}{4}$); 
\item for $b=-1$, $x=-1$ and all elements $x\in\F_{p^n}$ such that $\eta(x+1)=-1=-\eta(x)$ (that is, $|S_{-1,1}|=\frac{p^n-\eta(-1)}{4}=\frac{p^n-1}{4}$); 
\item one solution   if $\eta(b+1)=\eta(b-1)$, and no solution otherwise. 
\end{itemize}
We get then, for $b=\pm1$, $\frac{p^n+3}{4}$ solutions for~\eqref{eq:c-1}.

The spectrum is therefore
\allowdisplaybreaks
\begin{align*}
_c{\omega_0}&=|\{b\neq\pm1: \eta(b+1)\neq\eta(b-1)\}|\\
_c{\omega_1}&=|\{b\neq\pm1: \eta(b+1)=\eta(b-1)\}|\\
_c{\omega_{\frac{p^n+3}{4}}}&=|\{b=\pm1\}|=2.
\end{align*}

The $c$-differential uniformity is then $\frac{p^n+3}{4}$, achieved only in two points. 
We want to compute the cardinality of $T_{i,j}=\{b\neq\pm1: \eta(b-1)=i,\eta(b+1)=j\}$, $i,j\in\{\pm 1\}$. We have that
$$\sum_{b\in\F_{p^n}}(1+i\eta(b-1))(1+j\eta(b+1))=4|T_{i,j}|+2+j\eta(2)+i\eta(-2).$$
On the other hand,
$$\sum_{b\in\F_{p^n}}(1+i\eta(b-1))(1+j\eta(b+1))=p^n+ij\sum_{b\in\F_{p^n}}\eta((b-1)(b+1)).$$

We know \cite{LN97} that if $f(x)=a_2 x^2+a_1 x+a_0$ is a polynomial in a finite field $\F_q$ of odd characteristic, $a_2\neq 0, d=a_1^2-4a_0a_2$, and $\eta$ is the quadratic character on $\F_q$, then the Jacobsthal sum of $f$ is
 \begin{equation}
 \label{eq:Jacobsthal}
 \sum_{x\in\F_q} \eta(f(x))=\begin{cases}
 -\eta(a_2)&\text{ if }  d\neq 0\\
 (q-1)\eta(a_2)&\text{ if }  d= 0.
 \end{cases}
 \end{equation}
 To compute $\sum_{b\in\F_{p^n}}\eta((b-1)(b+1))$, we apply~\eqref{eq:Jacobsthal} and
 take $f(x)=x^2-1$, $d=1$, and so, the Jacobsthal sum becomes
 $ \sum_{x\in\F_{p^n}} \eta(x^2-1)=-\eta(1)=-1$. Therefore
\[
\sum_{b\in\F_{p^n}}(1+i\eta(b-1))(1+j\eta(b+1))=p^n-ij, 
\]
so 
\[
|T_{i,j}|=\frac{1}{4}\left(p^n-ij-2-j\eta(2)-i\eta(-2)\right).
\]

If $\eta(-1)=1$, then $|T_{1,1}|=\frac{1}{4}(p^n-3-2\eta(2))$, $|T_{1,-1}|=\frac{1}{4}(p^n-1)=|T_{-1,1}|$, and $|T_{-1,-1}|=\frac{1}{4}(p^n-3+2\eta(2))$ (%we can find $\eta(2)$, though, not needed, as we can see below, since we know
recall that $2$ is a quadratic residue if and only if $p \equiv \pm 1 \pmod 8$, or,  $p\equiv \pm 3\pmod 8$ and $n$ even).

%Here we will use that $\sum_{i=0}^{_c{\delta_F}}i {_c{\omega_i}}$. Then $$_c{\omega_1}+2\frac{p^n+3}{4}=p^n$$ This implies that $_c{\omega_1}=\frac{p^n-3}{2}$ (note that this can also be computed in a similar way as the cyclotomic numbers). The spectrum is then
Thus, 
\begin{align*}
_c{\omega_0}&=|\{b\neq\pm1: \eta(b+1)\neq\eta(b-1)\}|=|T_{1,-1}|+|T_{-1,1}|=\frac{p^n-1}{2},\\
_c{\omega_1}&=|\{b\neq\pm1: \eta(b+1)=\eta(b-1)\}|=|T_{1,1}|+|T_{-1,-1}|=\frac{p^n-3}{2},
\end{align*}
and thus, the spectrum becomes
\begin{equation}
\label{eq:spectrum3}
\left\{_c{\omega_0}=\frac{p^n-1}{2},\,_c{\omega_1}=\frac{p^n-3}{2},\,_c{\omega_{\frac{p^n+3}{4}}}=2\right\}.
\end{equation}

For  $\eta(-1)=-1$, we get, for $b=1$, $x=0$, $x=-1$, and all elements $x\in\F_{p^n}$ such that $\eta(x+1)=1=-\eta(x)$ (that is, $|S_{1,-1}|=\frac{p^n-2+\eta(-1)}{4}=\frac{p^n-3}{4}$); for $b=-1$, all elements $x\in\F_{p^n}$ such that $\eta(x+1)=-1=-\eta(x)$ (that is, $|S_{-1,1}|=\frac{p^n-\eta(-1)}{4}=\frac{p^n+1}{4}$); two solutions if $\eta(b+1)=\eta(b-1)=\eta(2)$, and no solution, otherwise. 
The spectrum is then:
\allowdisplaybreaks
\begin{align*}
_c{\omega_2}&=|\{b\neq\pm1: \eta(b+1)=\eta(b-1)=\eta(2)\}|\\
_c{\omega_0}&=|\{b\neq\pm1: b\notin c_{\omega_2}\}|\\
_c{\omega_{\frac{p^n+5}{4}}}&=|\{b=1\}|=1\\
_c{\omega_{\frac{p^n+1}{4}}}&=|\{b=-1\}|=1.
\end{align*}

If $\eta(-1)=-1$, then $|T_{1,1}|=\frac{1}{4}(p^n-3)=|T_{-1,-1}|$, $|T_{1,-1}|=\frac{1}{4}(p^n-1+2\eta(2))$ and $|T_{1,1}|=\frac{1}{4}(p^n-1-2\eta(2))$. Then:
\begin{align*}
_c{\omega_2}&=|\{b\neq\pm1: \eta(b+1)=\eta(b-1)=\eta(2)\}|=\frac{1}{4}(p^n-3),\\
_c{\omega_0}&=|\{b\neq\pm1: b\notin c_{\omega_2}\}|=p^n-2-\frac{1}{4}(p^n-3)=\frac{3p^n-5}{4}.
\end{align*}

The $c$-differential uniformity is then $\frac{p^n+5}{4}$, achieved only in one point, and the spectrum is then
%Here we will use that $\sum_{i=0}^{_c{\delta_F}}i {_c{\omega_i}}$. Then $$2_c{\omega_2}+\frac{p^n+5}{4}+\frac{p^n+1}{4}=p^n$$ This implies that $_c{\omega_2}=\frac{p^n-3}{4}$. The spectrum is then
\begin{equation}
\label{eq:spectrum4}
\left\{_c{\omega_0}=\frac{3p^n-5}{4},\,_c{\omega_2}=\frac{p^n-3}{4},\,_c{\omega_{\frac{p^n+1}{4}}}=1,\,_c{\omega_{\frac{p^n+5}{4}}}=1\right\}.
\end{equation}
The theorem is therefore shown.
\end{proof}

 Using our above counts and \cite[Theorem 4]{YZ22}, we have the following corollary:
\begin{corollary} The number of solutions $\left( {{x_1},{x_2},{x_3},{x_4}} \right) \in (\F_{p^n})^4$ of the following system of equations
	\begin{equation*}
		\begin{cases}
			{x_1} - {x_2} + {x_3} - {x_4} = 0\\
			x_1^d + x_2^d - x_3^d - x_4^d = 0
		\end{cases}
	\end{equation*}
is $\frac{1}{8}(p^{3n}+9p^{2n}-5p^n+3)$ when $\eta(-1)=1$, and is $\frac{1}{8}(p^{3n}+13p^{2n}-9p^n+3)$ when $\eta(-1)=-1$, where $d=\frac{p^n+1}{2}$.
\end{corollary}

We continue with the $c\neq 0, \pm1$ cases under $\eta(-1)=1$ (that is, $p\equiv 1\pmod 4$, or $p\equiv 3\pmod 4$ and $n$ even). For $c\in\F_{p^n}$ fixed, we will need the following elliptic curve $E(\F_{p^n}):\ y^2=x(x-1)(x-c^2)$ over the   field $\F_{p^n}$. 
By the Hasse-Weil inequality~\cite[Corollary 4.15]{Wash08}, we know that $|\#E(\F_{p^n})-(p^n+1)|\leq 2\sqrt{p^n}$. We let $a_{p^n}^{(c)}=\#E(\F_{p^n})-(p^n+1)$ denote the Frobenius trace (as it can be interpreted as the trace of the Frobenius endomorphism of the elliptic curve).
\begin{theorem} 
\label{thm:eta=1}
Let $F(x)=x^d$, where $d=\frac{p^n +1}{2}$, over $\F_{p^n}, n\geq 1$, $p$ odd prime. If $c\neq0,\pm1$ and $\eta(-1)=1$, then  the $c$-differential spectrum of $F$ is the following$:$
\begin{itemize}
\item[$(i)$]
If  $\eta(1-c^2)=1$$:$
\allowdisplaybreaks
\begin{align*}
_c{\omega_0}&=\frac{1}{8}\left(p^n+ a_{p^n}^{(c)}+5\right),\\
_c{\omega_1}&=\frac{1}{4}\left(3p^n-  a_{p^n}^{(c)}-5\right),\\
_c{\omega_2}&=\frac{1}{8}\left(p^n+ a_{p^n}^{(c)}+5\right).
\end{align*}

Note that this implies that the function is APcN for those values of $c$.

\item[$(ii)$]
If  $\eta(1-c)=1,\eta(1+c)=-1$$:$
\allowdisplaybreaks
\begin{align*}
_c{\omega_1}
&=\frac{1}{4}\left(p^n+ a_{p^n}^{(c)}+1\right)+\left\{\begin{array}{ll}
2,&\eta(2)=1,\eta(c)=1\\
1,&\eta(c)=-1\\
0,&\eta(2)=-1,\eta(c)=1,
\end{array}\right.\\
_c{\omega_2}&=\frac{1}{4} \left(p^n- a_{p^n}^{(c)}-7\right),\\
_c{\omega_3}&=\left\{\begin{array}{ll}
1,&\eta(c)=-1\\
0,&\eta(c)=1,
\end{array}\right.\\
_c{\omega_4}& =\frac{1}{16}\left(p^n+ a_{p^n}^{(c)}-3-4\eta(2)(1+\eta(c))\right),\\ 
_c{\omega_0}&=p^n-(_c{\omega_1}+_c{\omega_2}+_c{\omega_3}+_c{\omega_4}).
\end{align*}

\item[$(iii)$]
If  $\eta(1-c)=-1,\eta(1+c)=1$$:$
\allowdisplaybreaks
\begin{align*}
_c{\omega_1}&=\frac{1}{4}\left(p^n+  a_{p^n}^{(c)}+1\right)+\left\{\begin{array}{ll}
0,&\eta(2)=1,\eta(c)=1\\
1,&\eta(c)=-1\\
2,&\eta(2)=-1,\eta(c)=1,
\end{array}\right.\\
_c{\omega_2}&=\frac{1}{4}\left(p^n- a_{p^n}^{(c)}-7\right),\\
_c{\omega_3}&=\left\{\begin{array}{ll}
2,&\eta(2)=1,\eta(c)=1\\
1,&\eta(c)=-1\\
0,&\eta(2)=-1,\eta(c)=1,
\end{array}\right.\\
_c{\omega_4}&=\frac{1}{16}\left(p^n+ a_{p^n}^{(c)}-3+4\eta(2)(1+\eta(c))\right),\\
_c{\omega_0}&=p^n-\left(_c{\omega_1}+\,_c{\omega_2}+\,_c{\omega_3}+\,_c{\omega_4}\right).
\end{align*}
\end{itemize}

\end{theorem}
\begin{proof}
When   $c\neq 0, \pm1$, we analyze the possible solutions for~\eqref{eq:main}:

\begin{itemize}

\item If $x_1=0$ is a solution,  then $b=1$, otherwise, 0 is not a solution.

\item If $x_2=-1$ is a solution,  then $b/c=\eta(-1)$, otherwise $-1$ is not a solution. 

\item Let $\eta(x+1)=1=\eta(x)$. Then:
\begin{align*}
(1-c)x+1&=b,\text{ and so, } x_3=\frac{b-1}{1-c},\text{ thus},\\
\eta(x_3)&=\eta(b-1)\eta(1-c),\\
\eta(x_3+1)&=\eta(b-c)\eta(1-c).
\end{align*}

\item Let $\eta(x+1)=-1=\eta(x)$. Then:
\allowdisplaybreaks
\begin{align*}
-(1-c)x-1&=b ,\text{ and so, } x_4=-\frac{b+1}{1-c}\\
\eta(x_4)&=\eta(-1)\eta(b+1)\eta(1-c),\\
\eta(x_4+1)&=\eta(-1)\eta(b+c)\eta(1-c).
\end{align*}

\item Let $\eta(x+1)=1=-\eta(x)$. Then:
\allowdisplaybreaks
\begin{align*}
(1+c)x+1&=b, \text{ and so, }  x_5=\frac{b-1}{1+c},\\
\eta(x_5)&=\eta(b-1)\eta(1+c),\\
\eta(x_5+1)&=\eta(b+c)\eta(1+c).
\end{align*}

\item Let $\eta(x+1)=-1=-\eta(x)$. Then:
\allowdisplaybreaks
\begin{align*}
(-1-c)x-1&=b,\text{ and so, } x_6=\frac{b+1}{-1-c},\\
\eta(x_6)&=\eta(-1)\eta(b+1)\eta(1+c),\\
\eta(x_6+1)&=\eta(-1)\eta(b-c)\eta(1+c).
\end{align*}

\end{itemize}

%NB: DEPENDING ON THE CASE, 1 AND/OR -1 MIGHT BE EQUAL TO ONE OR MORE OF THE OTHER SOLUTIONS, INCREASING THE COUNT FOR THESE VALUES.

\noindent
{\bf Case A.} $\eta(1-c)=1=\eta(1+c)$   ($c\neq 0,\pm1$). Then:
\allowdisplaybreaks
\begin{align*}
\eta(x_3)&=\eta(b-1)=1, \eta(x_3+1)=\eta(b-c)=1,\\
\eta(x_4)&=\eta(-1)\eta(b+1)=-1, \eta(x_4+1)=\eta(-1)\eta(b+c)=-1,\\
\eta(x_5)&=\eta(b-1)=-1, \eta(x_5+1)=\eta(b+c)=1,\\
\eta(x_6)&=\eta(-1)\eta(b+1)=1, \eta(x_6+1)=\eta(-1)\eta(b-c)=-1.
\end{align*}

\noindent
{\bf Case A1.} $\eta(-1)=1$. Then:
\allowdisplaybreaks
\begin{align*}
 \eta(x_3)&=\eta(b-1)=1,  \eta(x_3+1)=\eta(b-c)=1,\\
\eta(x_4)&=\eta(b+1)=-1, \eta(x_4+1)=\eta(b+c)=-1,\\
\eta(x_5)&=\eta(b-1)=-1, \eta(x_5+1)=\eta(b+c)=1,\\
\eta(x_6)&=\eta(b+1)=1, \eta(x_6+1)=\eta(b-c)=-1.
\end{align*}

We see that, for a fixed $b$, if $x_3$ exists, then $x_5$ and $x_6$ cannot exist, and viceversa. Also, if $x_4$ exists, then $x_5$ and $x_6$ cannot exist, and viceversa. We can therefore have at most the solutions $x_3$ and $x_4$, or $x_5$ and $x_6$. 

NB: For $b=-1$, if $\eta(2)=1$ (that is,  $p \equiv \pm 1 \pmod 8$, or,  $p\equiv \pm 3\pmod 8$ and $n$ even), we get a solution $x_3$. If $\eta(2)=-1$  (that is,   $p\equiv \pm 1 \pmod 8$ or $n$ odd), we get a solution $x_5$. Regardless, we can see that $b=-1$ gives one solution. 
For $b=-c$, if $\eta(2c)=1$, we get a solution~$x_3$. If $\eta(2c)=-1$, we get a solution~$x_6$. Regardless, we can see that $b=-c$ gives one solution. 

For $c$ fixed, and $i,j,u,v\in\{\pm1 \}$, we define the following intersection of cyclotomic sets, 
\begin{align*}
S^c_{i,j,u,v}&=\{ b\neq\pm1,\pm c : (\eta(b-1),\eta(b+1),\eta(b-c),\eta(b+c))=(i,j,u,v)\}.
\end{align*}

We note the following distribution, which depends on the existence of such a~$b$ for each fixed $c$: 
\allowdisplaybreaks
\begin{align*}
_c{\omega_0}&=|S^c_{1,-1,-1,1}\cup S^c_{-1,1,1,-1}|,\\ 
_c{\omega_2}&=|S^c_{1,-1,1,-1}\cup S^c_{-1,1,-1,1}|,\\ 
_c{\omega_1}&=p^n-(_c{\omega_0}+\,_c{\omega_2}).
\end{align*}

\noindent
{\bf Case B.} $\eta(1-c)=-1=\eta(1+c)$ (thus,  $c\neq\pm1 $). In this case, all the above equations change sign, which implies the same compatibilities. The distribution is then given as above, with the values of $\eta(b-1),\eta(b+1),\eta(b-c),\eta(b+c)$   of opposite sign, giving:

\noindent
{\bf Case B1.} $\eta(-1)=1$: 
\allowdisplaybreaks
\begin{align*}
_c{\omega_0}&=|S^c_{1,-1,-1,1}\cup S^c_{-1,1,1,-1}|,\\ 
_c{\omega_2}&=|S^c_{1,-1,1,-1}\cup S^c_{-1,1,-1,1}|,\\ 
_c{\omega_1}&={p^n}-(_c{\omega_0}+\,_c{\omega_2}).
\end{align*}

\noindent {\bf Case C.} $\eta(1-c)=1=-\eta(1+c)$ (thus,  $c\neq\pm1 $):
\allowdisplaybreaks
\begin{align*}
\eta(x_3)&=\eta(b-1)=1, \eta(x_3+1)=\eta(b-c)=1,\\
\eta(x_4)&=\eta(-1)\eta(b+1)=-1, \eta(x_4+1)=\eta(-1)\eta(b+c)=-1,\\
\eta(x_5)&=-\eta(b-1)=-1,  \eta(x_5+1)=-\eta(b+c)=1,\\
\eta(x_6)&=-\eta(-1)\eta(b+1)=1, \eta(x_6+1)=-\eta(-1)\eta(b-c)=-1.
\end{align*}

\noindent {\bf Case C1.} $\eta(-1)=1$:
\allowdisplaybreaks
\begin{align*}
\eta(x_3)&=\eta(b-1)=1, \eta(x_3+1)=\eta(b-c)=1,\\
\eta(x_4)&=\eta(b+1)=-1, \eta(x_4+1)=\eta(b+c)=-1,\\
\eta(x_5)&=-\eta(b-1)=-1\Rightarrow\eta(b-1)=1,\\
\eta(x_5+1)&=-\eta(b+c)=1\Rightarrow\eta(b+c)=-1,\\
\eta(x_6)&=-\eta(b+1)=1\Rightarrow\eta(b+1)=-1,\\
\eta(x_6+1)&=-\eta(b-c)=-1\Rightarrow\eta(b-c)=1.
\end{align*}
We see that, for a fixed $b\neq\pm1,\pm c$, all four solutions are a priori possible. 

NB: For $b=1$, if $\eta(2)=-1$, we get the distinct solutions $x_4\neq0$ and $x_6\neq0$. In this case, we can see that $b=1$ gives three distinct solutions. 
For $b=c$, we get the distinct solutions $x_4,x_5$ if $\eta(2c)=-1$, and the solution $x_2$, giving three distinct solutions in that case.

The distribution in this case is: 
\begin{itemize}
\item
If $\eta(2)=1$ and $\eta(2c)=1$ (thus, $\eta(c)=1$): 
\allowdisplaybreaks
\begin{align*}
_c{\omega_1}&=|S^c_{1,1,1,1}\cup S^c_{1,1,-1,-1}\cup S^c_{-1,-1,1,1}\cup S^c_{-1-1,-1,-1}\cup\{b=1,c\}|,\\
_c{\omega_2}&=|S^c_{1,1,1,-1}\cup S^c_{1,-1,1,1}\cup S^c_{1,-1,-1,-1}\cup S^c_{-1,-1,1,-1}|,\\
_c{\omega_4}&=|S^c_{1,-1,1,-1}|,\\
_c{\omega_0}&={p^n}-(_c{\omega_1}+\,_c{\omega_2}+\,_c{\omega_4}).
\end{align*}

\item 
If $\eta(2)=1$ and $\eta(2c)=-1$ (thus, $\eta(c)=-1$): 
\allowdisplaybreaks
\begin{align*}
_c{\omega_1}&=|S^c_{1,1,1,1}\cup S^c_{1,1,-1,-1}\cup S^c_{-1,-1,1,1}\cup S^c_{-1-1,-1,-1}\cup\{b=1\}|,\\
_c{\omega_2}&=|S^c_{1,1,1,-1}\cup S^c_{1,-1,1,1}\cup S^c_{1,-1,-1,-1}\cup S^c_{-1,-1,1,-1}|,\\
_c{\omega_3}&=|\{b=c\}|,\\
_c{\omega_4}&=|S^c_{1,-1,1,-1}|,\\
_c{\omega_0}&={p^n}-\-(_c{\omega_1}+\, _c{\omega_2}+\,_c{\omega_3}+\, _c{\omega_4}).
\end{align*}

\item
If $\eta(2)=-1$ and $\eta(2c)=1$ (thus, $\eta(c)=-1$):
\allowdisplaybreaks
\begin{align*}
_c{\omega_1}&=|S^c_{1,1,1,1}\cup S^c_{1,1,-1,-1}\cup S^c_{-1,-1,1,1}\cup S^c_{-1-1,-1,-1}\cup\{b=c\}|,\\
_c{\omega_2}&=|S^c_{1,1,1,-1}\cup S^c_{1,-1,1,1}\cup S^c_{1,-1,-1,-1}\cup S^c_{-1,-1,1,-1}|,\\
_c{\omega_3}&=|\{b=1\}|,\\
_c{\omega_4}&=|S^c_{1,-1,1,-1}|,\\
_c{\omega_0}&={p^n}-(_c{\omega_1}+\, _c{\omega_2}+\,_c{\omega_3}+\, _c{\omega_4}).
\end{align*}

\item If $\eta(2)=-1$ and $\eta(2c)=-1$ (thus, $\eta(c)=1$):
\allowdisplaybreaks
\begin{align*}
_c{\omega_1}&=|S^c_{1,1,1,1}\cup S^c_{1,1,-1,-1}\cup S^c_{-1,-1,1,1}\cup S^c_{-1-1,-1,-1}|,\\
_c{\omega_2}&=|S^c_{1,1,1,-1}\cup S^c_{1,-1,1,1}\cup S^c_{1,-1,-1,-1}\cup S^c_{-1,-1,1,-1}|,\\
_c{\omega_3}&=|\{b=1,c\}|,\\
_c{\omega_4}&=|S^c_{1,-1,1,-1}|,\\
_c{\omega_0}&={p^n}-(_c{\omega_1}+\, _c{\omega_2}+\, _c{\omega_3}+\, _c{\omega_4}).
\end{align*}
\end{itemize}

\noindent {\bf Case D.} $\eta(1-c)=-1=-\eta(1+c) (\Rightarrow c\neq\pm1)$: In this case, all above equations change sign, which implies the same compatibilities.

\noindent {\bf Case D1.} $\eta(-1)=1$: 
\begin{itemize}
\item If $\eta(2)=1$ and $\eta(2c)=\eta(c)=1$:
\allowdisplaybreaks
\begin{align*}
_c{\omega_1}&=|S^c_{-1,-1,-1,-1}\cup S^c_{-1,-1,1,1}\cup S^c_{1,1,-1,-1}\cup S^c_{1,1,1,1}|,\\
_c{\omega_2}&=|S^c_{-1,-1,-1,1}\cup S^c_{-1,1,-1,-1}\cup S^c_{-1,1,1,1}\cup S^c_{1,1,-1,1}|,\\
_c{\omega_3}&=|\{b=1,c\}|,\\
_c{\omega_4}&=|S^c_{-1,1,-1,1}|,\\
_c{\omega_0}&={p^n}-(_c{\omega_1}+\,_c{\omega_2}+\,_c{\omega_3}+\,_c{\omega_4}).
\end{align*}
\item If $\eta(2)=1$ and $\eta(2c)=\eta(c)=-1$:
\allowdisplaybreaks
\begin{align*}
_c{\omega_1}&=|S^c_{-1,-1,-1,-1}\cup S^c_{-1,-1,1,1}\cup S^c_{1,1,-1,-1}\cup S^c_{1,1,1,1}\cup\{b=c\}|,\\
_c{\omega_2}&=|S^c_{-1,-1,-1,1}\cup S^c_{-1,1,-1,-1}\cup S^c_{-1,1,1,1}\cup S^c_{1,1,-1,1}|,\\
_c{\omega_3}&=|\{b=1\}|,\\
_c{\omega_4}&=|S^c_{-1,1,-1,1}|,\\
_c{\omega_0}&={p^n}-(_c{\omega_1}+\,_c{\omega_2}+\,_c{\omega_3}+\,_c{\omega_4}).
\end{align*}

\item If $\eta(2)=-1$ and $\eta(2c)=-\eta(c)=1$: 
\allowdisplaybreaks
\begin{align*}
_c{\omega_1}&=|S^c_{-1,-1,-1,-1}\cup S^c_{-1,-1,1,1}\cup S^c_{1,1,-1,-1}\cup S^c_{1,1,1,1}\cup\{b=1\}|,\\
_c{\omega_2}&=|S^c_{-1,-1,-1,1}\cup S^c_{-1,1,-1,-1}\cup S^c_{-1,1,1,1}\cup S^c_{1,1,-1,1}|,\\
_c{\omega_3}&=|\{b=c\}|,\\
_c{\omega_4}&=|S^c_{-1,1,-1,1}|,\\
_c{\omega_0}&={p^n}-(_c{\omega_1}+\,_c{\omega_2}+\,_c{\omega_4}).
\end{align*}
\item If $\eta(2)=-1$ and $\eta(2c)=-\eta(c)=-1$: 
\allowdisplaybreaks
\begin{align*}
_c{\omega_1}&=|S^c_{-1,-1,-1,-1}\cup S^c_{-1,-1,1,1}\cup S^c_{1,1,-1,-1}\cup S^c_{1,1,1,1}\cup\{b=1,c\}|,\\
_c{\omega_2}&=|S^c_{-1,-1,-1,1}\cup S^c_{-1,1,-1,-1}\cup S^c_{-1,1,1,1}\cup S^c_{1,1,-1,1}|,\\
_c{\omega_4}&=|S^c_{-1,1,-1,1}|,\\
_c{\omega_0}&={p^n}-(_c{\omega_1}+\,_c{\omega_2}+\,_c{\omega_3}+\,_c{\omega_4}).
\end{align*}
\end{itemize}

%Jacobsthal sums for the quartics:

For fixed $i,j,u,v\in\{\pm1 \}$, we will now compute the  cardinalities of the cyclotomic sets $S^c_{i,j,u,v}$, which, we recall are defined by 
\allowdisplaybreaks
\begin{align*}
S^c_{i,j,u,v}&=\{b\neq\pm1,\pm c : (\eta(b-1),\eta(b+1),\eta(b-c),\eta(b+c))=(i,j,u,v)\}. 
\end{align*}
For fixed $(i,j,u,v)$, we have
\allowdisplaybreaks
\begin{align*}
&\sum_{b\in\F_{p^n}} (1+i\eta(b-1))(1+j\eta(b+1))(1+u\eta(b-c))(1+v\eta(b+c))\\
=&16|S^c_{ijuv}|+(1+j\eta(2))(1+u\eta(1-c))(1+v\eta(1+c))\\
&+(1+i\eta(-2))(1+u\eta(-1-c))(1+v\eta(-1+c))\\
&+(1+i\eta(c-1))(1+j\eta(c+1))(1+v\eta(2c))\\
&+(1+i\eta(-c-1))(1+j\eta(-c+1))(1+u\eta(-2c)).
\end{align*}
On the other hand,
\allowdisplaybreaks
\begin{align*}
&\sum_{b\in\F_{p^n}} (1+i\eta(b-1))(1+j\eta(b+1))(1+u\eta(b-c))(1+v\eta(b+c))\\
=&\sum_{b\in\F_{p^n}}1+\sum_{x\in\F_{p^n}}i\eta(b-1)+\sum_{b\in\F_{p^n}}j\eta(b+1)+\sum_{b\in\F_{p^n}}u\eta(b-c)\\
&+\sum_{b\in\F_{p^n}}v\eta(b+c)+\sum_{b\in\F_{p^n}}i\eta(b-1)j\eta(b+1)+\sum_{b\in\F_{p^n}}i\eta(b-1)u\eta(b-c)\\
&+\sum_{b\in\F_{p^n}}i\eta(b-1)v\eta(b+c)+\sum_{b\in\F_{p^n}}j\eta(b+1)u\eta(b-c)\\
&+\sum_{b\in\F_{p^n}}j\eta(b+1)v\eta(b+c)+\sum_{b\in\F_{p^n}}u\eta(b-c)v\eta(b+c)\\
&+\sum_{b\in\F_{p^n}}i\eta(b-1)j\eta(b+1)u\eta(b-c)+\sum_{b\in\F_{p^n}}i\eta(b-1)j\eta(b+1)v\eta(b+c)\\
&+\sum_{b\in\F_{p^n}}i\eta(b-1)u\eta(b-c)v\eta(b+c)+\sum_{b\in\F_{p^n}}j\eta(b+1)u\eta(b-c)v\eta(b+c)\\
&+\sum_{b\in\F_{p^n}}i\eta(b-1)j\eta(b+1)u\eta(b-c)v\eta(b+c)\\
=&p^n+\sum_{b\in\F_{p^n}}ij\eta(b^2-1)+\sum_{b\in\F_{p^n}}iu\eta(b^2-(1+c)b+c)\\
&+\sum_{b\in\F_{p^n}}iv\eta(b^2+(-1+c)b-c)+\sum_{b\in\F_{p^n}}ju\eta(b^2+(1-c)b-c)\\
&+\sum_{b\in\F_{p^n}}jv\eta(b^2+(1+c)b+c)+\sum_{b\in\F_{p^n}}uv\eta(b^2-c^2)\\
&+\sum_{b\in\F_{p^n}}iju\eta((b^2-1)(b-c))+\sum_{b\in\F_{p^n}}ijv\eta((b^2-1)(b+c))\\
&+\sum_{b\in\F_{p^n}}iuv\eta((b-1)(b^2-c^2)+\sum_{b\in\F_{p^n}}juv\eta((b+1)(b^2-c^2))\\
&+\sum_{b\in\F_{p^n}}ijuv\eta((b^2-1)(b^2-c^2)).
\end{align*}
The cubics and the quartics will be more complicated to deal with. As for the quadratics (taking into account that $c\neq\pm1,0$), and using Equation~\eqref{eq:Jacobsthal}, which gives the Jacobsthal sum in terms of the discriminant $d$  of the polynomial $f$, we obtain
\allowdisplaybreaks
\begin{align*}
f(b)&=b^2-1, d=4\neq0,\\
& \text{ thus, } \sum_{b\in\F_{p^n}}ij\eta(b^2-1)=-ij\eta(1)=-ij,\\
f(b)&=b^2-(1+c)b+c, d=(1+c)^2-4c=c^2-2c+1=(c-1)^2\neq0 \\
&\text{ thus, }\sum_{b\in\F_{p^n}}iu\eta(b^2-(1+c)b+c)=-iu\eta(1)=-iu,\\
f(b)&=b^2+(-1+c)b-c, d=(-1+c)^2+4c=c^2+2c+1=(c+1)^2\neq0, \\
&\text{ thus, }\sum_{b\in\F_{p^n}}iv\eta(b^2+(-1+c)b-c)=-iv\eta(1)=-iv,\\
f(b)&=b^2+(1-c)b-c, d=(1-c)^2+4c=c^2+2c+1=(c+1)^2\neq0, \\
&\text{ thus, }\sum_{b\in\F_{p^n}}ju\eta(b^2+(-1+c)b-c)=-ju\eta(1)=-ju,\\
f(b)&=b^2+(1+c)b+c, d=(1+c)^2-4c=c^2-2c+1=(c-1)^2\neq0, \\
&\text{ thus, }\sum_{b\in\F_{p^n}}jv\eta(b^2+(-1+c)b-c)=-jv\eta(1)=-jv,\\
f(b)&=b^2-c^2, d=4c^2\neq0,\\
&\text{ thus, }\sum_{b\in\F_{p^n}}uv\eta(b^2-c^2)=-uv\eta(1)=-uv.
\end{align*}

Summarizing, we obtain,
\allowdisplaybreaks
\begin{align*}
&\sum_{b\in\F_{p^n}} (1+i\eta(b-1))(1+j\eta(b+1))(1+u\eta(b-c))(1+v\eta(b+c))\\
=&p^n-(ij+iu+iv+ju+jv+uv)\\
&+ij\left(u\sum_{b\in\F_{p^n}}\eta((b^2-1)(b-c))+v\sum_{b\in\F_{p^n}}\eta((b^2-1)(b+c))\right)\\
&+uv\left(i\sum_{b\in\F_{p^n}}\eta((b-1)(b^2-c^2)+j\sum_{b\in\F_{p^n}}\eta((b+1)(b^2-c^2))\right)\\
&+ijuv\sum_{b\in\F_{p^n}}\eta((b^2-1)(b^2-c^2)).
\end{align*}

We can further simplify the sums. If $\eta(-1)=1$, then (changing $b$ to $-b$), we get 
\begin{align*}
\sum_{b\in\F_{p^n}}\eta((b^2-1)(b-c))&=\sum_{b\in\F_{p^n}}\eta((b^2-1)(b+c))\\
 \sum_{b\in\F_{p^n}}\eta((b-1)(b^2-c^2)&=\sum_{b\in\F_{p^n}}\eta((b+1)(b^2-c^2)).
 \end{align*}
If $\eta(-1)=-1$, then 
\allowdisplaybreaks
\begin{align*}
\sum_{b\in\F_{p^n}}\eta((b^2-1)(b-c))&=-\sum_{b\in\F_{p^n}}\eta((b^2-1)(b+c))\\ 
 \sum_{b\in\F_{p^n}}\eta((b-1)(b^2-c^2)&=-\sum_{b\in\F_{p^n}}\eta((b+1)(b^2-c^2)).
 \end{align*}
Therefore, if $i=-\eta(-1)j$ and/or $u=-\eta(-1)v$, there are several simplifications. We denote by $A=\sum_{b\in\F_{p^n}}\eta((b^2-1)(b-c))$, $B=\sum_{b\in\F_{p^n}}\eta((b-1)(b^2-c^2)$ and $C=\sum_{b\in\F_{p^n}} \eta((b^2-1)(b^2-c^2))$.

If $\eta(-1)=1$, we note that, since the map $b\mapsto -b$ also maps $\eta(b-1)\mapsto\eta(b+1)$ and $\eta(b-c)\mapsto\eta(b+c)$, then $|S^c_{1,-1,1,1}|=|S^c_{-1,1,1,1}|$, $|S^c_{1,-1,1,-1}|=|S^c_{-1,1,-1,1}|$, $|S^c_{1,1,1,-1}|=|S^c_{1,1,-1,1}|$ and $|S^c_{1,-1,-1,1}|=|S^c_{-1,1,1,-1}|$. The cardinality of these cyclotomic sets is given here in lexicographic order:
\allowdisplaybreaks
\begin{align*}
|S^c_{1,1,1,1}|&=\frac{1}{16}\left(p^n-6+2A+2B+C\right.\\
&\left.\qquad\quad-2(2+\eta(2)+\eta(2c))(1+\eta(1-c))(1+\eta(1+c))\right),\\
|S^c_{1,1,1,-1}|&=|S^c_{1,1,-1,1}|=\frac{1}{16}\left(p^n-2B-C-2(1+\eta(2))(1-\eta(1-c^2))\right.\\
&\left.\qquad\quad -2(1+\eta(1-c))(1+\eta(1+c))\right),\\
|S^c_{1,1,-1,-1}|&=\frac{1}{16}(p^n+2-2A+2B+C\\
&\qquad\quad-2(1+\eta(2))(1-\eta(1-c))(1-\eta(1+c))\\ 
&\qquad\quad-2(1+\eta(1-c))(1+\eta(1+c))(1-\eta(2c))),\\ 
|S^c_{1,-1,1,1}|&=\frac{1}{16}(p^n-2A-C-2(1+\eta(1-c))(1+\eta(1+c))\\
&\qquad\quad-2(1+\eta(2c))(1-\eta(1-c^2))),\\
|S^c_{1,-1,1,-1}|&=\frac{1}{16}(p^n+2+C\\
&\qquad\quad-(2-\eta(2)-\eta(2c))(1+\eta(1-c))(1-\eta(1+c))\\
&\qquad\quad-(2+\eta(2)+\eta(2c))(1-\eta(1-c))(1+\eta(1+c))),\\
|S^c_{1,-1,-1,1}|&=\frac{1}{16}(p^n+2+C\\
&\qquad\quad-(2-\eta(2)-\eta(2c))(1-\eta(1-c))(1+\eta(1+c))\\
&\qquad\quad-(2+\eta(2)+\eta(2c))(1+\eta(1-c))(1-\eta(1+c))),\\
|S^c_{1,-1,-1,-1}|&=\frac{1}{16}(p^n+2A-C-2(1-\eta(1-c))(1-\eta(1+c))\\
&\qquad\quad-2(1-\eta(2c))(1-\eta(1-c^2))),\\
|S^c_{-1,1,1,1}|&=\frac{1}{16}(p^n-2A-C-2(1+\eta(1-c))(1+\eta(1+c))\\
&\qquad\quad-2(1+\eta(2c))(1-\eta(1-c^2))),\\
|S^c_{-1,1,1,-1}|&=\frac{1}{16}(p^n+2+C\\
&\qquad\quad-(2+\eta(2)+\eta(2c))(1+\eta(1-c))(1-\eta(1-c))\\
&\qquad\quad-(2-\eta(2)-\eta(2c))(1-\eta(1-c))(1+\eta(1-c))),\\
|S^c_{-1,1,-1,1}|&=\frac{1}{16}(p^n+2+C\\
&\qquad\quad-(2+\eta(2)+\eta(2c))(1-\eta(1-c))(1+\eta(1-c))\\
&\qquad\quad-(2-\eta(2)-\eta(2c))(1+\eta(1-c))(1-\eta(1-c))),\\
|S^c_{-1,1,-1,-1}|&=\frac{1}{16}(p^n+2A-C-2(1-\eta(1-c))(1-\eta(1+c))\\
&\qquad\quad-2(1-\eta(2c))(1-\eta(1-c^2))),\\
|S^c_{-1,-1,1,1}|&=\frac{1}{16}(p^n+2+2A-2B+C\\
&\qquad\quad-2(1-\eta(2))(1+\eta(1-c))(1+\eta(1+c))\\
&\qquad\quad-2(1-\eta(c-1))(1-\eta(c+1))(1+\eta(2c))),\\
|S^c_{-1,-1,1,-1}|&=|S^c_{-1,-1,-1,1}|=\frac{1}{16}(p^n+2B-C\\
&\qquad\quad-2(1-\eta(2))(1-\eta(1-c^2))-2(1-\eta(1-c))(1-\eta(1+c))),\\
|S^c_{-1,-1,-1,-1}|&=\frac{1}{16}(p^n-6-2A-2B+C\\
&\qquad\quad-2(2-\eta(2)-\eta(2c))(1-\eta(1-c))(1-\eta(1+c))).
\end{align*}

%Just the sum without terms:
%
%$$|S^c_{1,-1,1,-1}|=|S^c_{1,-1,-1,1}|=|S^c_{-1,1,1,-1}|=|S^c_{-1,1,-1,1}|=\frac{1}{16}(p^n+2+C)$$
%$$|S^c_{1,1,1,-1}|=|S^c_{1,1,-1,1}|=\frac{1}{16}(p^n-2B-C)$$
%$$|S^c_{-1,-1,1,-1}|=|S^c_{-1,-1,-1,1}|=\frac{1}{16}(p^n+2B-C)$$
%$$|S^c_{1,-1,1,1}|=|S^c_{-1,1,1,1}|=\frac{1}{16}(p^n-2A-C)$$
%$$|S^c_{1,-1,-1,-1}|=|S^c_{-1,1,-1,-1}|=\frac{1}{16}(p^n+2A-C)$$
%$$|S^c_{1,1,1,1}|=\frac{1}{16}(p^n-6+2A+2B+C)$$
%$$|S^c_{-1,-1,-1,-1}|=\frac{1}{16}(p^n-6-2A-2B+C)$$
%$$|S^c_{-1,-1,1,1}|=\frac{1}{16}(p^n+2+2A-2B+C)$$
%$$|S^c_{1,1,-1,-1}|=\frac{1}{16}(p^n+2-2A+2B+C)$$

%Note that the symmetries in the sets and the spectra imply that we do not need the values of $A$ and $B$.

For Cases A1 and B1 (that is, when $\eta(-1)=1$ and $\eta(1-c^2)=1$), the spectrum is 
(recall that $C=\sum_{b\in\F_{p^n}} \eta((b^2-1)(b^2-c^2))$)
\allowdisplaybreaks
\begin{align*}
_c{\omega_0}&=|S^c_{1,-1,-1,1}|+|S^c_{-1,1,1,-1}|=\frac{1}{8}(p^n+2+C),\\
_c{\omega_2}&=|S^c_{1,-1,1,-1}|+S^c_{-1,1,-1,1}|=\frac{1}{8}(p^n+2+C),\\
_c{\omega_1}&=p^n-(_c{\omega_0}+\,_c{\omega_2})=\frac{1}{4}(3p^n-2-C).
\end{align*}

For Case C1 (that is, $\eta(-1)=1$, $\eta(1-c)=1,\eta(1+c)=-1$): 
\allowdisplaybreaks
\begin{align*}
_c{\omega_1}&=|S^c_{1,1,1,1}|+|S^c_{1,1,-1,-1}|+|S^c_{-1,-1,1,1}|+|S^c_{-1,-1,-1,-1}|\\
&\qquad\qquad\qquad\qquad \qquad\qquad\qquad  +\left\{\begin{array}{ll}
2,&\eta(2)=1,\eta(c)=1\\
1,&\eta(c)=-1\\
0,&\eta(2)=-1,\eta(c)=1
\end{array}\right.\\
&=\frac{1}{4}(p^n-2+C)+\left\{\begin{array}{ll}
2,&\eta(2)=1,\eta(c)=1\\
1,&\eta(c)=-1\\
0,&\eta(2)=-1,\eta(c)=1,
\end{array}\right.\\
_c{\omega_2}&=|S^c_{1,1,1,-1}|+|S^c_{1,-1,1,1}|+|S^c_{1,-1,-1,-1}|+|S^c_{-1,-1,1,-1}|=\frac{1}{4}(p^n-C-4),\\
_c{\omega_3}&=\left\{\begin{array}{ll}
1,&\eta(c)=-1\\
0,&\eta(c)=1,
\end{array}\right.\\
_c{\omega_4}&=|S^c_{1,-1,1,-1}|=\frac{1}{16}(p^n+2+C-4(2+\eta(2)+\eta(2c))),\\
_c{\omega_0}&=p^n-(_c{\omega_1}+\,_c{\omega_2}+\,_c{\omega_3}+\,_c{\omega_4}).
\end{align*}

For the case D1 (that is, $\eta(-1)=1$, $\eta(1-c)=-1,\eta(1+c)=1$):
\allowdisplaybreaks
\begin{align*}
_c{\omega_1}&=|S^c_{1,1,1,1}|+|S^c_{1,1,-1,-1}|+|S^c_{-1,-1,1,1}|+|S^c_{-1,-1,-1,-1}|\\
&\qquad\qquad\qquad\qquad \qquad\qquad\qquad   +\left\{\begin{array}{ll}
0,&\eta(2)=1,\eta(c)=1\\
1,&\eta(c)=-1\\
2,&\eta(2)=-1,\eta(2c)=-1,
\end{array}\right.\\
&=\frac{1}{4}(p^n-2+C)+\left\{\begin{array}{ll}
0,&\eta(2)=1,\eta(2c)=1\\
1,&\eta(c)=-1\\
2,&\eta(2)=-1,\eta(2c)=-1,
\end{array}\right.\\
_c{\omega_2}&=|S^c_{-1,-1,-1,1}|+|S^c_{-1,1,-1,-1}|+|S^c_{-1,1,1,1}|+|S^c_{1,1,-1,1}|=\frac{1}{4}(p^n-C-4),\\
_c{\omega_3}&=\left\{\begin{array}{ll}
2,&\eta(2)=1,\eta(c)=1\\
1,&\eta(c)=-1\\
0,&\eta(2)=-1,\eta(c)=1,
\end{array}\right.\\
_c{\omega_4}&=|S^c_{-1,1,-1,1}|=\frac{1}{16}(p^n+2+C-4(2-\eta(2)-\eta(2c))),\\
_c{\omega_0}&=p^n-(_c{\omega_1}+\,_c{\omega_2}+\,_c{\omega_3}+\,_c{\omega_4}).
\end{align*}
%Note that $A$ and $B$ are not relevant for the spectrum.
Note that the symmetries in the sets and the spectra imply that we do not need the values of $A$ and $B$.

We can furthermore transform the quartic $C$ into the sum a quadratic and a cubic. Using the labelling $b^2=a$, we get the sums
\allowdisplaybreaks
\begin{align*}
&\sum_{b\in\F_{p^n}} \eta((b^2-1)(b^2-c^2)) = \sum_{a\in\F_{p^n}} (\eta(a)+1) \eta((a-1)(a-c^2))\\
&= \sum_{a\in\F_{p^n}}  \eta(a(a-1)(a-c^2))+\sum_{a\in\F_{p^n}}  \eta((a-1)(a-c^2))\\
&= \sum_{a\in\F_{p^n}} \eta(a(a-1)(a-c^2))-1,
\end{align*}
where we used Jacobsthal's sum~\eqref{eq:Jacobsthal} for the second sum.

In order to compute the sum corresponding to the above cubic, we count points on elliptic curves.
We could consider several particular cases for $c\neq 0,\pm 1$ (like $c^2\in\F_p$), but we prefer simplicity (though, we make some remark after the proof of the theorem on that, as well as provide a result for $c^2=-1$ in Theorem~\ref{thm:c-1}).

In order to find 
$\sum_{a\in\F_{p^n}} \eta(a(a-1)(a-c^2))$, we consider the elliptic curve $E(\F_{p^n}):\ y^2=x(x-1)(x-c^2)$ over the   field $\F_{p^n}$. 
By the Hasse-Weil inequality~\cite[Corollary 4.15]{Wash08}, we know that $|\#E(\F_{p^n})-(p^n+1)|\leq 2\sqrt{p^n}$.
If we denote by $a_{p^n}^{(c)}=\#E(\F_{p^n})-(p^n+1)$, the Hasse-Weil inequality error, then  
\allowdisplaybreaks
\begin{align*}
C+1&=\sum_{a\in\F_{p^n}} \eta(a(a-1)(a-c^2))\\
&= \#\{ a\neq 0,1,c^2\,:\, a(a-1)(a-c^2)\text{ is a square}\}\\
&\qquad -\#\{ a\neq 0,1,c^2\,:\, a(a-1)(a-c^2)\text{ is a non-square}\}\\
&= 2\#\{ a\neq 0,1,c^2\,:\, a(a-1)(a-c^2)\text{ is a square}\}-p^n+3\\
&= \#E(\F_{p^n})-p^n+3=a_{p^n}^{(c)}+4,
\end{align*}
and so, $C=  a_{p^n}^{(c)}+3$.

Substituting this in the spectra expressions given above, we get the claimed spectra and the theorem is shown. 
\end{proof}

We can express the  sum $C$ with respect to a smaller field, namely $\F_q=\F_p(c^2)$, where $q=p^r$, for some positive integer $r$, in the following way. First, we recall a (particular case of a) celebrated result of Weil and generalized by Deligne~(see~\cite[Theorem 4.12]{Wash08}, for more details).  Given an elliptic curve in the Weierstrass form $y^2=x^3+Ax+B$ over some finite field $\F_{q}$ ($q$ is a prime power), if the number of rational  points of the elliptic curve over the base field $\F_q$ is $\#E(\F_q)=q+1-a_q$, the Frobenius polynomial of the curve is $T^2-a_q T+q=(T-\alpha)(T-\beta)$ (for some complex numbers $\alpha,\beta$), and the number of  points of the same curve over the extension field $\F_{q^n}$ is $\#E(\F_{q^n})=(1-\alpha^n)(1-\beta^n) =1+q^n-(\alpha^n+\beta^n)$, for all $n\geq 1$. We thus get $a_{p^n}^{(c)}=\alpha^n+\beta^n$, where 
$\alpha=\frac{1}{2}\left( a_q^{(c)}+\sqrt{\left(a_q^{(c)}\right)^2-4q}\right)$, 
$\beta=\frac{1}{2}\left( a_q^{(c)}-\sqrt{\left(a_q^{(c)}\right)^2-4q}\right)$. Observe that when $c^2\in \F_p$, then $q=p$. 

We will consider the special case of $c^2=-1$ (and so, $p\equiv 1\pmod 4$, or, $p\equiv 3\pmod 4$ and $n$ even) below.

%Recall that $2$ is a quadratic residue if and only if $p \equiv \pm 1 \pmod 8$, or,  $p\equiv \pm 3\pmod 8$ and $n$ even.

\begin{theorem}
\label{thm:c-1}
Let $p$ be a prime number such that either $p\equiv 1\pmod 4$, or $p\equiv 3\pmod 4$ and $n$ even. Let $c$ be a square root of $-1$ in $\F_{p^n}$. The $c$-differential spectrum of $x\mapsto x^{\frac{p^n+1}{2}}$ on  $\F_{p^n}$ is given by the following$:$
\allowdisplaybreaks
\begin{itemize}
\item[$(i)$]  If $n$ is even and either  $p\equiv 3\pmod 8$ or $p\equiv 7\pmod 8$ (that is, $p\equiv 1\pmod 4$ and $n$ even), then
% 2 is a square, $a_p^{(c)}=0$ and $a_{p^n}^{(c)}=2(-p)^{\frac{n}{2}}$
\allowdisplaybreaks
\begin{align*}
_c{\omega_0}&=\frac{1}{8}\left(p^n+2(-p)^{\frac{n}{2}}+5\right),\\
_c{\omega_1}&=\frac{1}{4}\left(3p^n-2(-p)^{\frac{n}{2}}-5\right),\\
_c{\omega_2}&=\frac{1}{8}\left(p^n+2(-p)^{\frac{n}{2}}+5\right).
\end{align*}
Note that this implies that the map is APcN.

\item[$(ii)$]   Let $p\equiv 1\pmod 4$, and write $p=a^2+b^2$, $a,b$ integers, $b$ even, and $a+b\equiv 1\pmod 4$.
\begin{enumerate}
\item[$(1)$]  If $p\equiv 1\pmod8$ or $p\equiv 5\pmod 8$ and   $n$ even, then
\allowdisplaybreaks
\begin{align*}
_c{\omega_0}&=\frac{1}{8}\left(p^n+(a+bc)^n+(a-bc)^n+5\right),\\
_c{\omega_1}&=\frac{1}{4}\left(3p^n-(a+bc)^n-(a-bc)^n-5\right),\\
_c{\omega_2}&=\frac{1}{8}\left(p^n+(a+bc)^n+(a-bc)^n+5\right).
\end{align*}
Note that this implies that the map is APcN.

\item[$(2)$]  
If $p\equiv 5\pmod 8$ and $n$ odd, then 
\allowdisplaybreaks
\begin{align*}
_c{\omega_0} &=\frac{1}{16}\left(7p^n-  (a+bc)^n-(a-bc)^n-5\right),\\
_c{\omega_1}
&=\frac{1}{4}\left(p^n+ (a+bc)^n+(a-bc)^n+5\right),\\
_c{\omega_2}&=\frac{1}{4} \left(p^n-(a+bc)^n-(a-bc)^n-7\right),\\
_c{\omega_3}&=1,\\
_c{\omega_4}& =\frac{1}{16}\left(p^n+(a+bc)^n+(a-bc)^n-3 \right).
\end{align*}
%\item[$(b)$]  If  $\eta(1-c)=-1,\eta(1+c)=1$, then
%\allowdisplaybreaks
%\begin{align*}
%_c{\omega_0} &=\frac{1}{16}\left(7p^n- (a+bc)^n-(a-bc)^n+10\right),\\
%_c{\omega_1}&=\frac{1}{4}\left(p^n+(a+bc)^n+(a-bc)^n+5\right),\\
%_c{\omega_2}&=\frac{1}{4}\left(p^n- (a+bc)^n-(a-bc)^n-7\right),\\
%_c{\omega_3}&=1,\\
%_c{\omega_4}&=\frac{1}{16}\left(p^n+(a+bc)^n+(a-bc)^n-3\right).
%\end{align*}
% \end{enumerate}
\end{enumerate}
\end{itemize}
\end{theorem}
\begin{proof}
We use Theorem~\ref{thm:eta=1} and observe that the curve corresponding to the sum $C$ is now $y^2=x^3-x$ over $\F_p$. 
We will need a result~\cite[Theorem 4.23]{Wash08} dealing with the number of points of the well-known (since Gauss' time) elliptic curve $y^2=x^3-kx$ over $\F_p$, $k\not\equiv 0\pmod p$. If $p\equiv 3\pmod 4$, then $a_p=0$; if $p\equiv 1\pmod 4$, writing $p=a^2+b^2$, $a,b$ integers, $b$ even and $a+b\equiv 1\pmod 4$, then 
\begin{align*}
a_p=\begin{cases}
2a & \text{ if $k$ is a $4$th power mod $p$}\\
- 2a & \text{ if $k$ is a square but not a $4$th power mod $p$}\\
\pm 2b & \text{ if $k$ is not a $4$th power mod $p$}.
\end{cases}
\end{align*}
 In our case $k=1$, and the first case above applies when $p\equiv 1\pmod 4$. We will go through several cases, depending on the residue of $p\pmod 8$, since we know that 
 $2$ is a quadratic residue if and only if $p\equiv \pm 1\pmod 8$, or $n$ even. %Note that we need 
We will use below the observation that when $c^2=-1$, then $1+c=-c^2+c=c(1-c)$ and so, $\eta(1+c)=\eta(c)\eta(1-c)$, and so, $\eta(1+c),\eta(1-c)$ have opposite signs if and only if $\eta(c)=-1$. For simplicity, we write in the proof below $a_p, a_{p^n}$, in lieu of $a_p^{(c)},a_{p^n}^{(c)}$, and will all refer to the curve $y^2=x^3-x$ over $\F_p$, respectively, $\F_{p^n}$.

\noindent
{\bf Case 1.}
Let $n$ be even and  $p\equiv 3\pmod 8$ or $p\equiv 7\pmod 8$. 
%2 is a square, $a_p=0$
Since $a_p=0$ this case, then the Frobenius polynomial of the curve $y^2=x^3-x$ becomes $T^2+p=(T-\sqrt{-p})(T+\sqrt{-p})$, which renders $a_{p^n}=2(-p)^{\frac{n}{2}}$.
We now apply Theorem~\ref{thm:eta=1}~$(i)$, obtaining claim $(i)$ of our theorem.

\noindent
{\bf Case 2.}
Let $p\equiv 1\pmod 8$ or $p\equiv 5\pmod 8$ and $n$ even. Let  $p=a^2+b^2$, where $b$ is even.  %2 is a square, $a_p=2a$
In this case, the Frobenius polynomials for $y^2=x^3-x$ is $T^2-2aT+p=(T-\alpha)(T-\beta)$, where $\alpha=a+b c, \beta=a-b c$ (recall that $c$ is a root of $-1$). Therefore, $a_{p^n}=(a+bc)^n+(a-bc)^n$. We thus get claim $(ii)\,(1)$, using Theorem~\ref{thm:eta=1}~$(i)$.

\noindent
{\bf Case 3.}
Let $p\equiv 5\pmod 8$ and $n$ odd, $p=a^2+b^2$, $b$ even. 
%2 is not a square, $a_p=2a$
As above, $a_{p^n}=(a+bc)^n+(a-bc)^n$. We know that $2$ is not a square residue. We therefore get $(ii)\, (2)$, via Theorem~\ref{thm:eta=1}~$(ii/iii)$ (note that both cases give the same spectrum in this case, due to $\eta(c)=-1$).

The theorem is shown.
\end{proof}
 
%Recall $2$ is a quadratic residue if and only if $p \equiv \pm 1 \pmod 8$, or,  $p\equiv \pm 3\pmod 8$ and $n$ even.

We continue with the case of $c\neq 0,\pm 1$, $\eta(-1)=-1$ (that is,  $p\equiv 3\pmod 4$ and $n$ odd).
\begin{theorem} 
Let $F(x)=x^d$, where $d=\frac{p^n +1}{2}$ ($p$ odd prime, $n\geq 1$). If $c\neq0,\pm1$ and $\eta(-1)=-1$, then,    
\begin{itemize}
\item[$(i)$]
If $\eta(-1)=-1$ and $\eta(1-c^2)=1$$:$
\allowdisplaybreaks
\begin{align*}
_c{\omega_0}&=\frac{1}{8}\left(3p^n+a_{p^n}^{(c)}-3\right),\\
_c{\omega_1}&=\frac{1}{4}\left(p^n-a_{p^n}^{(c)}+1\right),\\
_c{\omega_2}&=\frac{1}{8}\left(3p^n+a_{p^n}^{(c)}+1\right).
 \end{align*}
\item[$(ii)$]
If $\eta(-1)=-1$ and $\eta(1-c)=1,\eta(1+c)=-1$$:$
\allowdisplaybreaks
\begin{align*}
_c{\omega_0}&=\frac{1}{8}\left(3p^n+a_{p^n}^{(c)}-21\right),\\
_c{\omega_1}&=\frac{1}{4}\left(p^n-a_{p^n}^{(c)}+1\right),\\
_c{\omega_2}&=\frac{1}{8}\left(3p^n+a_{p^n}^{(c)}+19\right).
 \end{align*}
 
\item[$(iii)$]
If $\eta(-1)=-1$ and $\eta(1-c)=-1,\eta(1+c)=1$$:$ 
\allowdisplaybreaks
\begin{align*}
_c{\omega_0}&=\frac{1}{8}\left(3p^n+a_{p^n}^{(c)}+3+4\eta(2)\right),\\
 _c{\omega_1}&=\frac{1}{4}\left(p^n-a_{p^n}^{(c)}+1\right),\\
 _c{\omega_2}&=\frac{1}{8}\left(3p^n+a_{p^n}^{(c)}-5-4\eta(2)\right).
 \end{align*}
\end{itemize}
\end{theorem}
\begin{proof}
%Similar for $\eta(-1)=-1$: 
\noindent {\bf Case A2.} $\eta(-1)=-1$, $\eta(1-c)=1,\eta(1+c)=1$. Then:
\allowdisplaybreaks
\begin{align*}
\eta(x_3)&=\eta(b-1)=1, \eta(x_3+1)=\eta(b-c)=1,\\
\eta(x_4)&=\eta(b+1)=1, \eta(x_4+1)=\eta(b+c)=1,\\
\eta(x_5)&=\eta(b-1)=-1, \eta(x_5+1)=\eta(b+c)=1,\\
\eta(x_6)&=\eta(b+1)=-1, \eta(x_6+1)=\eta(b-c)=1.
\end{align*}

We see that, for a fixed $b$, if $x_3$ exists, then $x_5$ cannot exist, and vice-versa. Also, if $x_4$ exists, then $x_6$ cannot exist, and vice-versa. We can at most have $x_3$ and $x_4$, or $x_3$ and  $x_6$, or $x_4$ and $x_5$, or  $x_5$ and $x_6$. 

NB: For $b=1$, if $\eta(2)=1$ (that is, $p\equiv \pm 1\pmod 8$, or,  $p\equiv \pm 3\pmod 8$ and $n$ even), we get a solution $x_4\neq0$. If $\eta(2)=-1$ (that is,   $p\equiv \pm 3\pmod 8$ and $n$ odd), we get a solution $x_6\neq0$. Regardless, we can see that $b=1$ gives two distinct solutions.
For $b=c$, if $\eta(2c)=1$, we get two distinct solutions, $x_4, x_5$. 

The distribution in this case is: 
\begin{itemize}
\item If $\eta(2c)=1$, 
\allowdisplaybreaks
\begin{align*}
_c{\omega_1}&=|S^c_{1,1,1,-1}\cup S^c_{1,1,-1,1}\cup S^c_{-1,-1,1,-1}\cup S^c_{-1,-1,-1,1}\cup\{b=-c\}|,\\
_c{\omega_2}&=|S^c_{1,1,1,1}\cup S^c_{1,-1,1,1}\cup S^c_{1,-1,1,-1}\cup S^c_{-1,1,1,1}\\
&\qquad\qquad\qquad\qquad  \cup S^c_{-1,1,-1,1}\cup S^c_{-1,-1,1,1}\cup\{b=1,c\}|,\\
_c{\omega_0}&=p^n-(_c{\omega_1}+\,_c{\omega_2}).
\end{align*}
\item
If $\eta(2c)=-1$,
\allowdisplaybreaks
\begin{align*}
_c{\omega_1}&=|S^c_{1,1,1,-1}\cup S^c_{1,1,-1,1}\cup S^c_{-1,-1,1,-1}\cup S^c_{-1,-1,-1,1}\cup\{b=-c\}|,\\
_c{\omega_2}&=|S^c_{1,1,1,1}\cup S^c_{1,-1,1,1}\cup S^c_{1,-1,1,-1}\cup S^c_{-1,1,1,1}\\
&\qquad\qquad\qquad\qquad \cup S^c_{-1,1,-1,1}\cup S^c_{-1,-1,1,1}\cup\{b=1\}|,\\
_c{\omega_0}&={p^n}-(_c{\omega_1}+\,_c{\omega_2}).
\end{align*}
\end{itemize}
Using similar techniques as in the case $\eta(-1)=1$, we see that the spectrum for A2 is
\allowdisplaybreaks
\begin{align*}
_c{\omega_1}&=\frac{1}{4}(p^n-C)+1=\frac{1}{4}(p^n-C+4),\\
_c{\omega_2}&=\frac{1}{8}(3p^n+C-2),\\
_c{\omega_0}&=\frac{1}{8}(3p^n+C-6).
\end{align*}
NB: The cases $\eta(2c)=1$ and $\eta(2c)=-1$ yield the same result.

\noindent {\bf Case B2.}
 $\eta(-1)=-1$, $\eta(1-c)=-1,\eta(1+c)=-1$. Then:
\begin{itemize}
\item If $\eta(2c)=1$:
\allowdisplaybreaks
\begin{align*}
_c{\omega_1}&=|S^c_{-1,-1,-1,1}\cup S^c_{-1,-1,1,-1}\cup S^c_{1,1,-1,1}\cup S^c_{1,1,1,-1}\cup\{b=-c\}|,\\
_c{\omega_2}&=|S^c_{-1,-1,-1,-1}\cup S^c_{-1,1,-1,-1}\cup S^c_{-1,1,-1,1}\cup S^c_{1,-1,-1,-1}\\
&\qquad\qquad\qquad\qquad \cup S^c_{1,-1,1,-1}\cup S^c_{1,1,-1,-1}\cup\{b=1\}|,\\
_c{\omega_0}&={p^n}-(_c{\omega_1}+\,_c{\omega_2}).
\end{align*}

\item 
If $\eta(2c)=-1$:
\allowdisplaybreaks
\begin{align*}
_c{\omega_1}&=|S^c_{-1,-1,-1,1}\cup S^c_{-1,-1,1,-1}\cup S^c_{1,1,-1,1}\cup S^c_{1,1,1,-1}\cup\{b=-c\}|,\\
_c{\omega_2}=&|S^c_{-1,-1,-1,-1}\cup S^c_{-1,1,-1,-1}\cup S^c_{-1,1,-1,1}\cup S^c_{1,-1,-1,-1}\\
&\qquad\qquad\qquad\qquad \cup S^c_{1,-1,1,-1}\cup S^c_{1,1,-1,-1}\cup\{b=1,c\}|,\\
_c{\omega_0}&={p^n}-(_c{\omega_1}+\,_c{\omega_2}).
\end{align*}
\end{itemize}

In this case, then, we obtain the same spectrum as for A2, that is 
\allowdisplaybreaks
\begin{align*}
_c{\omega_1}&= \frac{1}{4}(p^n-C+4),\\
_c{\omega_2}&=\frac{1}{8}(3p^n+C-2),\\
_c{\omega_0}&=\frac{1}{8}(3p^n+C-6).
\end{align*}

\noindent {\bf Case C2.}  $\eta(-1)=-1$, $\eta(1-c)=1,\eta(1+c)=-1$. Then:
\allowdisplaybreaks
\begin{align*}
\eta(x_3)&=\eta(b-1)=1, \eta(x_3+1)=\eta(b-c)=1,\\
\eta(x_4)&=\eta(b+1)=1, \eta(x_4+1)=\eta(b+c)=1,\\
\eta(x_5)&=-\eta(b-1)=-1\Rightarrow\eta(b-1)=1,\\
\eta(x_5+1)&=-\eta(b+c)=1\Rightarrow\eta(b+c)=-1,\\
\eta(x_6)&=-\eta(b+1)=-1\Rightarrow\eta(b+1)=1,\\
\eta(x_6+1)&=-\eta(b-c)=1\Rightarrow\eta(b-c)=-1.
\end{align*}

We see that, for a fixed $b$, if $x_3$ exists, then $x_6$ cannot exist, and viceversa. Also, if $x_4$ exists, then $x_5$ cannot exist, and viceversa. We can at most have $x_3$ and $x_4$, or $x_3$ and $x_5$, or $x_4$ and $x_6$, or  $x_5$ and $x_6$. 

NB: For $b=-1$, we can see that, if $\eta(2)=-1$, we get the distinct solutions $x_3$ and $x_5$. In this case, we can see that $b=-1$ gives two distinct solutions. For $b=-c$, if $\eta(2c)=1$, we get the solution $x_6\neq-1$. If $\eta(2c)=-1$, we get the solution $x_3\neq-1$.  We can see that $b=-c$ always gives two distinct solutions.

The distribution in this case is:
\begin{itemize}
\item 
If $\eta(2)=1$:
\allowdisplaybreaks
\begin{align*}
_c{\omega_0}&=|S^c_{1,-1,-1,1}\cup S^c_{-1,1,1,-1}\cup S^c_{-1,-1,1,1}\cup S^c_{-1,-1,1,-1}\\
&\qquad\qquad\qquad \cup S^c_{-1,-1,-1,1} \cup S^c_{-1,-1,-1,-1}\cup\{b=-1,c\}|,\\
_c{\omega_1}&=|S^c_{1,-1,1,1}\cup S^c_{1,-1,-1,-1}\cup S^c_{-1,1,1,1}\cup S^c_{-1,1,-1,-1}\cup\{b=1\}|,\\
_c{\omega_2}&={p^n}-(_c{\omega_0}+\, _c{\omega_2}).
\end{align*}

\item 
If $\eta(2)=-1$:
\allowdisplaybreaks
\begin{align*}
_c{\omega_0}&=|S^c_{1,-1,-1,1}\cup S^c_{-1,1,1,-1}\cup S^c_{-1,-1,1,1}\cup S^c_{-1,-1,1,-1}\\
&\qquad\qquad\qquad\qquad \cup S^c_{-1,-1,-1,1}\cup S^c_{-1,-1,-1,-1}\cup\{b=c\}|,\\
_c{\omega_1}&=|S^c_{1,-1,1,1}\cup S^c_{1,-1,-1,-1}\cup S^c_{-1,1,1,1}\cup S^c_{-1,1,-1,-1}\cup\{b=1\}|,\\
_c{\omega_2}&={p^n}-(_c{\omega_0}+\, _c{\omega_2}).
\end{align*}
\end{itemize}

Here the spectrum is therefore
%If $\eta(2)=1$:
\allowdisplaybreaks
\begin{align*}
_c{\omega_0}&=\frac{1}{8}(3p^n-40+C)+2=\frac{1}{8}(3p^n+C-24),\\
_c{\omega_1}&=\frac{1}{4}(p^n-C)+1=\frac{1}{4}(p^n-C+4),\\
 _c{\omega_2}&=\frac{1}{8}(5p^n+40+C)-3=\frac{1}{8}(3p^n+C+16).
 \end{align*}

\noindent {\bf Case D2.}  $\eta(-1)=-1$, $\eta(1-c)=-1,\eta(1+c)=1$. Then:
\allowdisplaybreaks
\begin{align*}
_c{\omega_0}&=|S^c_{-1,1,1,-1}\cup S^c_{1,-1,-1,1}\cup S^c_{1,1,-1,-1}\cup S^c_{1,1,-1,1}\\
&\qquad\qquad\qquad\qquad\cup S^c_{1,1,1,-1}\cup S^c_{1,1,1,1}\cup\{b=-1,c\}|\\
_c{\omega_1}&=|S^c_{-1,1,-1,-1}\cup S^c_{-1,1,1,1}\cup S^c_{1,-1,-1,-1}\cup S^c_{1,-1,1,1}\cup\{b=1\}|,\\
_c{\omega_2}&={p^n}-(_c{\omega_0}+\, _c{\omega_2}).
\end{align*}
%Note that this implies that, for $c\neq0,\pm1$,  the $c$-differential uniformity is upper bounded by four, and that it is upper bounded by two unless $\eta(-1)=1$ and there exist $b,c$ such that:
%$\eta(1-c)=\pm1=-\eta(1+c)$ and $(\eta(b-1),\eta(b+1),\eta(b-c),\eta(b+c))=(\pm1,\mp1,\pm1,\mp1)$. 

Here the spectrum is 
\allowdisplaybreaks
\begin{align*}
_c{\omega_0}&=\frac{1}{8}(3p^n+C-16+8\eta(2))+2=\frac{1}{8}(3p^n+C+4\eta(2)),\\
_c{\omega_1}&=\frac{1}{4}(p^n-C)+1=\frac{1}{4}(p^n-C+4),\\
_c{\omega_2}&=\frac{1}{8}(5p^n+C+16-8\eta(2))-3=\frac{1}{8}(3p^n+C-8-4\eta(2)).
\end{align*}
Setting the value of $C=  a_{p^n}^{(c)}+3$ as before, we obtain the claimed results.
\end{proof}

\section{Conclusions}

In this paper we concentrate on the map $x\mapsto x^{\frac{p^n+1}{2}}$ on a finite field $\F_{p^n}$, $p$ being an odd prime and $n$ a positive integer, and compute the complete $c$-differential spectrum, by using Weil sums and elliptic curve methods. Perhaps the methods we used can be applied to finding the spectrum of other interesting maps, with potential low differential or $c$-differential uniformity.

\vspace{.3cm}

\noindent
{\bf Acknowledgements.}
The authors would like to thank Tor Helleseth and Daniel Katz for interesting discussions and for comments on a previous draft.

%\vspace{.6cm}
%
%\noindent {\bf Ethical Approval and Consent to participate}\newline
%All authors entered the research freely  with full information about what it means for them to take part, and that they give consent before they entered the research.

%\vspace{.3cm}
%
%\noindent {\bf  Consent for publication}\newline
%All authors agreed to the submission.
%
%\vspace{.3cm}
%
%\noindent {\bf  Availability of supporting data}\newline
%There is no supporting data.
%
%\vspace{.3cm}
%
%\noindent {\bf   Competing interests}\newline
%There are no competing interests in this research.
%
%\vspace{.3cm}
%
%\noindent {\bf  Funding}\newline
%The research has not been supported by outside funding.
%
%\vspace{.3cm}
%
%\noindent {\bf  Authors' contributions}\newline
%All authors contributed equally to the research.

\end{document}